\documentclass[letterpaper, 10 pt, conference]{ieeeconf}
\usepackage[T1]{fontenc}% optional T1 font encoding

\IEEEoverridecommandlockouts                          
\usepackage{amsmath} % assumes amsmath package installed
\usepackage{amssymb}  % assumes amsmath package installed
\usepackage{xcolor}
\usepackage{caption}
\usepackage{graphicx}
\usepackage{mathtools}
\usepackage{caption}
\usepackage{floatrow}
\usepackage{tikz}
\usepackage{tikz-cd}

\usepackage{algorithmicx}
    \usepackage[ruled]{algorithm}
    \usepackage{algpseudocode}

\usepackage{pgfplots}
\usepackage{xcolor}
\usetikzlibrary{matrix,arrows,calc,positioning,shapes,decorations.pathreplacing}
\usepackage{graphicx}

\usepackage{amsmath} % assumes amsmath package installed

\usepackage{enumerate}
\usepackage[all,tips]{xy}
\SelectTips{cm}{11}

\newtheorem{theorem}{Theorem}

\newtheorem{remark}{Remark}

\newcommand{\R}{\mathbb{R}}

\begin{document}

\title{
Hamel equations and quasivelocities for\\ nonholonomic systems with
inequality constraints}

\author{Alexandre Anahory Simoes and Leonardo Colombo 
\thanks{A. Anahory Simoes (alexandre.anahory@ie.edu) is with the School of Science and Technology, IE University, Spain.}
\thanks{L. Colombo (leonardo.colombo@car.upm-csic.es) is with Centre for Automation and Robotics (CSIC-UPM), Ctra. M300 Campo Real, Km 0,200, Arganda
del Rey - 28500 Madrid, Spain.}  \thanks{The authors acknowledge financial support from Grant PID2019-106715GB-C21 funded by MCIN/AEI/ 10.13039/501100011033.}}%

\maketitle

\begin{abstract}
In this paper we derive Hamel equations for the motion of nonholonomic systems subject to inequality constraints in quasivelocities. As examples, the vertical rolling disk hitting a wall  and the Chaplygin sleigh with a knife edge constraint hitting a circular table are shown to illustrate the theoretical results.
\end{abstract}

%\begin{IEEEkeywords} Symmetry reduction, multi-agent systems, optimal control, collision and obstacle avoidance. \end{IEEEkeywords}

\section{Introduction}
Quasivelocities are the components of velocities, describing
a mechanical system, relative to a set of vector fields
(in principle, local) that span at each point the fiber of the
tangent bundle of the configuration space. The main point is
that these vector fields need not be associated with (local)
configuration coordinates on the configuration space. One
of the reasons for using quasivelocities is that the Euler-
Lagrange equations written in generalized coordinates are not
always effective for analyzing the dynamics of a mechanical
system of interest as it was shown in \cite{quasi}.

Some mechanical systems have a restriction either on the configurations that the system may assume or at the velocities the system is allowed to go. Systems with such restrictions are generally called constrained systems. Nonholonomic  systems  \cite{Bloch,generalized,dLMdD1996,Neimark} are, roughly speaking, mechanical systems with constraints on their velocity that are not derivable from position constraints. They arise, for instance, in mechanical systems that have rolling contact (e.g., the rolling of wheels without slipping) or certain kinds of sliding contact (such as the sliding of skates).  

 We will restrict ourselves to the case of linear constraints on the velocities, where the velocity lies in a subspace of the tangent space. The collection of this subspaces forms a distribution, denoted by $\mathcal{D}$, and is locally given by an expression of the type $\mu^{a}_{i}\dot{q}^{i}=0$. The nonholonomic equations of motion are obtained from the Lagrange-d'Alembert principle \cite{Bloch} and its local expression is
\begin{equation*}
    \begin{split}
        & \frac{d}{dt}\frac{\partial L}{\partial \dot{q}^{i}} - \frac{\partial L}{\partial q^{i}}= \lambda_{a}\mu^{a}_{i}, \\
        & \dot{q} \in \mathcal{D}_{q(t)}
    \end{split}
\end{equation*}
where $\lambda_{a}$ is a Lagrange multiplier that might be computed using the constraints.

%There are some differences between nonholonomic systems and classical Hamiltonian or Lagrangian systems. Among them, nonholonomic systems are nonvariational, they arise from Lagrange-d’Alembert principle and not from Hamilton’s principle; they may preserve the energy of the system as Hamiltonian systems, but they are not, in general, time-reversible, and they do not preserve, in general, the momentum for systems with symmetries (i.e., Noether’s theorem does not apply, in general, for nonholonomic systems). Nonholonomic systems are not Poisson, the volume is not preserved in phase space (allowing asymptotic behavior despite energy is conserved), but they are almost-Poisson (i.e., there is a bracket that together with the energy on phase space defines the motion, but the bracket does not satisfy, in general, the Jacobi identity). 
 
Mechanical systems subject to inequality constraints are confined within a region of space with boundary. Collision with the boundary activates constraint forces forbiding the system to cross the boundary into a non-admissible region of space (see, e.g., \cite{Fetecau, kaufman}). Inequality constraints appear, for instance, in the problem of rigid-body collisions, mechanical grasping models and biomechanical locomotion \cite{simoes,lopez}. Mechanical systems with impulsive effects on quasivelocities has been studied in \cite{leohybrid}. The dynamics for systems modeled by using quasivelocities
is governed by Hamel equations \cite{quasi}, \cite{Bloch}. In this paper, we
introduce Hamel equations for nonholonomic systems subject to inequality constraints building on the work \cite{quasi}.
A first approach to the dynamics of nonholonomic systems with inequality constraints was given in \cite{Clark}, using Weierstrass-Erdemann conditions to obtain the state of the system immediately after the collision, and later in \cite{alex}, where the authors use a variational principle to obtain both the equations of motion and the collision equations. In this paper, we go a step further and consider the quasivelocities description of nonholonomic systems and the corresponding Hamel equations generalizing the results of \cite{alex} to an adated basis of vector fields to the nonholonomic distribution defining the nonholonomic constraints.

The remainder of the paper is structured as folows. In Section \ref{sec2} we introduce mechanical systems with inequality constraints and the main notation used along the paper.  Section \ref{sec3} is devoted to study Hamel equations for systems with inequality constraints. We extend the analyis to nonholonomic systems with inequality constraints in Section \ref{sec4}. As examples, in Sections \ref{sec5} and \ref{sec6}  the vertical rolling disk hitting a wall  and the Chaplygin sleigh with a knife edge constraint hitting a circular table, respectively, are shown to illustrate the theoretical results.

\section{Mechanical systems with inequality constraints}\label{sec2}

Suppose $Q$ is a differentiable manifold of dimension $n$. Throughout the text, $q^{i}$ will denote a particular choice of local coordinates on this manifold and $TQ$ denotes its tangent bundle, with $T_{q}Q$ denoting the tangent space at a specific point $q\in Q$ generated by the coordinate vectors $\frac{\partial}{\partial q^{i}}$. Usually $v_{q}$ denotes a vector at $T_{q}Q$ and, in addition, the coordinate chart $q^{i}$ induces a natural coordinate chart on $TQ$ denoted by $(q^{i},\dot{q}^{i})$. There is a canonical projection $\tau_{Q}:TQ \rightarrow Q$, sending each vector $v_{q}$ to the corresponding base point $q$. Note that in coordinates $\tau_{Q}(q^{i},\dot{q}^{i})=q^{i}$. 

 The vertical lift of a vector field $X\in \mathfrak{X}(Q)$ to $TQ$ is defined by $$X_{v_{q}}^{V}=\left. \frac{d}{dt}\right|_{t=0} (v_{q} + t X(q)).$$

$T_{q}Q$ has a vector space structure, so we may consider its dual space, $T^{*}_{q}Q$ and define the cotangent bundle as $\displaystyle{T^{*}Q:=\bigcup_{q\in Q}T^{*}_{q}Q},$ with local coordinates $(q^i,p_i)$. 
%\subsection{Normal cone}

In this paper, we will analyse the dynamics of nonholonomic systems evolving on the configuration manifold $Q$ which are subjected to inequality constraints, i.e., constraints determined by a submanifold with boundary $C$ of the manifold $Q$. The boundary $\partial C$ is a smooth manifold of $Q$ with codimension $1$. Locally, the boundary $\partial C$ is a smooth manifold of the type
$\partial C = \{ q\in Q \ | \ g(q) = 0 \}$
and the manifold $C$ is
$C = \{ q\in Q \ | \ g(q) \leqslant 0 \}$
for some smooth function $g:Q \rightarrow \mathbb{R}$.

In convex geometry, given a closed convex set $K$ of $\mathbb{R}^{n}$, the \textit{polar cone} of $K$ is the set
$K^{p} = \{z\in \mathbb{R}^{n} \ | \ \langle z, y \rangle \leqslant 0, \forall y \in K \}$ (see \cite{brogliato} for instance).
The \textit{normal cone} to $K$ at a point $x\in K$ is given by
$N_{K}(x) = K^{p} \cap \{x\}^{T}$, 
where $\{x\}^{T}$ is the orthogonal subspace to $x$ with respect to the Euclidean inner product. Based on this construction, we will only use a minimal definition of normal cone suiting the kind of inequality constraints we will be dealing with. Given a submanifold with boundary $C$ as before, the normal cone to a point $q\in \partial C$ is the set
$N_{C}(q)=\{ \lambda dg(q) | \lambda \geqslant 0\}$.
The two definitions match, if $C$ is a closed convex set of $\R^{n}$ with boundary being a hypersurface of dimension $n-1$.

%\subsection{The jump equations}

Given a Lagrangian function $L:TQ\to\mathbb{R}$ describing the dynamics of a mechanical system, with local coordinates $(q^i,\dot{q}^i)$, $i=1,\ldots,n=\dim Q$, the equations of motion under the presence of  inequality constraints are given by Euler-Lagrange equations
$$\displaystyle{\frac{d}{dt}\frac{\partial L}{\partial \dot{q}^{i}} - \frac{\partial L}{\partial q^{i}}= 0}$$
whenever the trajectory is in the interior of the constraint submanifold $C\setminus \partial C$. At impact times $t_{i}\in \mathbb{R}$ of the trajectory with the boundary $q(t_{i})\in \partial C$, there is a discontinuity in the state variables of the system, often called a jump. This jump is determined by the equations:
\begin{equation}\label{jump:equations}
    \begin{split}
        & \frac{\partial L}{\partial \dot{q}}|_{t=t_{i}^{+}} - \frac{\partial L}{\partial \dot{q}}|_{t=t_{i}^{-}} \in -N_{C}, \,\,\, E_{L}|_{t=t_{i}^{+}} = E_{L}|_{t=t_{i}^{-}}.
    \end{split}
\end{equation}

\begin{remark}
    We note that a negative sign in the previous equation appears as a consequence of the non-interpenetrability of the constraint.i.e., the mechanical system may not cross the boundary of the admissible variational constraint. We will see exactly how the negative signs appears in the following section.
\end{remark}

Throuhgout the paper, $L$ will be a regular mechanical Lagrangian, i.e., it has the form kinetic minus potential energy \cite{Bloch} and the Legendre transform $\mathbb{F}L:TQ\rightarrow T^{*}Q$ is a local diffeomorphism.

	\section{Hamel's equations for systems with inequality constraints}\label{sec3}
\subsection{Hamel's equations}

In this section we briefly discuss the Hamel equations. The exposition follows paper \cite{quasi}.

In many cases the Lagrangian and the equations of motion of a mechanical system have a
simpler structure when these are written using velocity components measured against a frame that is unrelated to system's local configuration coordinates. Let $q = (q^1,\dots, q^n)$ be local coordinates on the configuration space $Q$ and $u_i \in TQ$, $i = 1, \dots,n$, be smooth independent \emph{local} vector fields defined in the same coordinate neighborhood (in certain cases, some or all of $u_i$ can be chosen to be \emph{global} vector fields on $Q$).
The components of $u_i$ relative to the basis $\partial /\partial q^j$ will be denoted $\psi_i ^j$; that is,
\[
%\begin{equation}\label{psidef}
u_i(q) = \psi _i^j(q) \frac{\partial}{\partial q^j},
%\end{equation}
\]
where $i, j = 1,\dots,n$ and where summation on $j$ is understood by employing Einstein summation notation.

Let $v = (v^1,\dots, v^n) \in \mathbb{R}^n$ be the components of the velocity vector $\dot q \in TQ$ relative to the basis $u_1, \dots, u_n$, \emph{i.e.}, 
\begin{equation}\label{noncommuting.variables.eqn}
\dot q = v^i u_i(q);
\end{equation}
then
\begin{equation}\label{l_u.eqn}
\ell(q, v) := L(q, v^i u_i (q))
\end{equation}
is the Lagrangian of the system written in the adapted coordinates $(q, v)$ on the tangent bundle $TQ$. The coordinates $(q, v)$ are Lagrangian analogues of non-canonical variables in Hamiltonian dynamics. 

Define the quantities $c_{ij}^m (q)$ by the equations
\begin{equation}\label{ccoeff.eqn}
[u_i (q),u_j (q)] = c^m_{ij} (q) u_m (q),
\end{equation} where 
$i, j, m = 1, \dots, n$. These quantities vanish if and only if the vector fields $u_i (q)$, $i = 1,\dots, n$, commute. Here and elsewhere, $[\cdot,\! \cdot]:\mathbb{R}^{m}\times\mathbb{R}^{m}\to\mathbb{R}^{m}$ is the Jacobi--Lie bracket of vector fields on $Q$.  Also one can find that $$c_{ij}^{m}=(\psi^{-1})_{k}^{m}\left(\frac{\partial\psi_{j}^{k}}{\partial q^{l}}\psi_{i}^{l}-\frac{\partial\psi_{i}^{k}}{\partial q^{l}}\psi_{j}^{l}\right).$$ The dual of $[\cdot,\cdot]_{q}$ is defined by the operation $[\cdot,\cdot]_{q}^{*}:V_q\times V_q^{*}\to V^{*}_{q}$ given by $$\langle[v,\alpha]_{q}^{*},w\rangle\equiv\langle ad_{v}^{*}\alpha,w\rangle:=\langle\alpha,[v,w]_{q}\rangle$$ where $V_{q}$ is the Lie algebra given by $V_{q}=(\mathbb{R}^{m},[\cdot,\cdot]_{q})$. Here $ad^{*}$ is the dual of  the usual \textit{ad} operator in a Lie algebra.

Viewing $u_i$ as vector fields on $TQ$ whose fiber components equal $0$ (that is, taking the vertical lift of these vector fields), one defines the directional derivatives $u_i [\ell]$ for a function $\ell :TQ\to \mathbb{R}$ by the formula 
\[
u_i [\ell] = \psi_i^j \frac{\partial \ell}{\partial q^j}.
\]

The evolution of the variables $(q,v)$ is governed by the \textit{Hamel equations} 
\begin{equation}\label{hamel.eqn}
\frac{d}{dt} \frac{\partial \ell}{\partial v^j}
= c^m_{ij}v^i \frac{\partial \ell}{\partial v^m} + u_j[\ell],
\end{equation} 
coupled with equations \eqref{noncommuting.variables.eqn}.
If $u_i = \partial / \partial q^i$, equations \eqref{hamel.eqn} become the Euler--Lagrange equations. Equations \eqref{hamel.eqn} were introduced in 
\cite{Ha1904} (see also \cite{Neimark} for details and some history).  Hamel equations can be written as $$\frac{d}{dt}\frac{\partial\ell}{\partial v}=\left[v,\frac{\partial\ell}{\partial v}\right]_{q}^{*}+u[\ell]\equiv ad_{v}^{*}\frac{\partial\ell}{\partial v}+u[\ell]$$ coupled with the equation $\dot{q}=v^{i}u_{i}(q)$.

\subsection{The jump equations in quasivelocities}

To obtain the jump equations in terms of quasivelocities, we generalize an extended variational principle derived in \cite{alex} for nonholonomic systems, which in turn is the nonholonomic version of the variational principle introduced in \cite{Fetecau} to obtain the equations satisfied by a system without constraints after a collision with a smooth submanifold in the configuration space.

\begin{theorem}
    Let $q:[0,h]\rightarrow Q$ and $v:[0,h]\rightarrow TQ$ be trajectories of the Hamel's equations for the Lagrangian function $\ell:TQ\rightarrow \mathbb{R}
    $ subjected to the inequality constraint $q(t)\in C$. Suppose that this system has an impact against the boundary $\partial C$ at the time $t_{k}\in [0,h]$. Then the trajectory satisfies Hamel's equations \eqref{hamel.eqn} in the intervals $[0, t_{k}^{-}[$ and $]t_{k}^{+},h]$ and at the impact time $t_{k}$, the following conditions hold:
    \begin{equation}\label{nh:jump:equations}
        \begin{split}
            & \frac{\partial\ell}{\partial v}|_{t=t_{k}^{+}} - \frac{\partial\ell}{\partial v}|_{t=t_{k}^{-}} \in -N_{C}, \\
            & E_{\ell}|_{t=t_{k}^{+}} = E_{\ell}|_{t=t_{k}^{+}},
        \end{split}
    \end{equation}
    where $E_{\ell}:TQ \rightarrow \mathbb{R}$ is the energy of the system given in local coordinates by $E_{\ell}(q,v)=\frac{\partial\ell}{\partial v^{i}} v^i - \ell (q,v)$.
\end{theorem}

\begin{proof} The curve $(q(t), v(t))$ is a critical point of the functional 
\begin{equation} \label{hamel_action.eqn}
\int_a^b \ell(q,v)\,dt
\end{equation} 
with respect to variations $\delta v$, induced by the variations $\delta q =  w^i u_i (q)$, and given by (see  \cite{quasi})
\begin{equation}\label{xi_var.eqn}
\delta v^k = \dot w^k + c_{ij}^k (q)v^i w^j.
\end{equation}
So, 
\begin{align*}&\delta \int_a^{t_{k}^{-}} \ell(q,v)\,dt +  \delta \int_{t_{k}^{+}}^{b} \ell(q,v)\,dt\\& = \int_a^{t_{k}^{-}} \left( c_{ij}^k v^{i}\frac{\partial\ell}{\partial v^{k}}+\psi^{i}_{j}\frac{\partial\ell}{\partial q^{i}} - \frac{d}{dt}\frac{\partial\ell}{\partial v^{j}}\right) w^j \ dt\\
&+ \int_{t_{k}^{+}}^{b} \left( c_{ij}^k v^{i}\frac{\partial\ell}{\partial v^{k}}+\psi^{i}_{j}\frac{\partial\ell}{\partial q^{i}} - \frac{d}{dt}\frac{\partial\ell}{\partial v^{j}}\right) w^j \ dt\\& - \left[ \frac{\partial\ell}{\partial v^{j}} w^j + \ell \delta t_{k}\right]_{t_{k}^{-}}^{t_{k}^{+}}\end{align*}

The jump condition follows from the fact that $q(t_{k})\in \partial C$ from where
$$\delta (q(t_{k})) \in T(\partial C) \implies \delta q (t_{k}) + \dot{q}(t_{k})\delta t_{k} \in T(\partial C).$$
In quasivelocities this condition becomes
$$w^i(t_{k}) u_i (q(t_{k})) + v^i(t_{k}) u_i(q(t_{k})) \delta t_{k} \in T(\partial C)$$

The variations satisfying the previous equation are spanned by variations $w^{i} (t_{k})u_{i}(q(t_{k})) \in T(\partial C)$ and $\delta t_{k} = 0$ or $\delta t_{k} = 1$ and $w^{i} (t_{k}) = - v^{i}(t_{k})$. From the latter we immediately deduce that
$$\left[ \frac{\partial\ell}{\partial v^{i}} v^i - \ell \right]_{t_{k}^{-}}^{t_{k}^{+}} = 0,$$
which is the energy conservation condition in the jump equations. From $\delta t_{k}=0$, we get that 
$$\frac{\partial\ell}{\partial v}|_{t=t_{k}^{+}} - \frac{\partial\ell}{\partial v}|_{t=t_{k}^{-}},$$
annihilates $\delta q = w^i u_i (q) \in T(\partial C)$.\end{proof}

\section{Nonholonomic systems with inequality constraints}\label{sec4}

Assume that there are velocity constraints imposed on the system. We will restrict to constraints that are linear in the velocities. Consider a distribution $\mathcal{D}$ on the configuration space $Q$ describing these constraints, that is, $\mathcal{D}$ is a collection of linear subspaces of $TQ$ ($\mathcal{D}_q\subset T_{q}Q$ for each $q\in Q$). A curve $q(t)\in Q$ will be said to satisfy the constraints if $\dot{q}(t)\in\mathcal{D}_{q(t)}$ for all $t$. Locally, the constraint distribution can be written as
$$\mathcal{D}=\{\dot{q}\in TQ\,|\,\mu_{i}^{a}(q)\dot{q}^{i}=0,\,\,a=1,\ldots,m\}.$$

% Constraints are nonholonomic if and only if they cannot be rewritten as position constraints. The distribution $\D$, in general, be nonintegrable, so the constraints are, in general, nonholonomic.   

The Lagrange-d'Alembert equations of motion for the system are those determined by $\delta\int_{a}^{b}L(q,\dot{q})dt=0,$ where we choose variations $\delta q(t)$ of the curve $q(t)$ that satisfy $\delta q(a)=\delta(b)=0$ and $\delta q(t)\in\mathcal{D}_{q(t)}$ for each $t\in[a,b]$. Note that here the curve $q(t)$ itself satisfies the constraints. Variations are taken before imposing the constraints and hence, the constraints are not imposed on the family of curves defining the variations.

%In order to describe the dynamical equations one can use the Euler-Lagrange equations with Lagrange multipliers.

%The constrained variations $\delta q(t)\in TQ$ satisfy the equations $\mu_{i}^{a}(q)\delta q^{i}=0,\quad a=1,\ldots,m$. From the constrained variations by applying the virtual displacement principle, the equations of motion for a nonholonomic system are given by $$\frac{d}{dt}\frac{\partial L}{\partial\dot{q}^{i}}=\frac{\partial L}{\partial q^{i}}+\lambda_{a}\mu_{i}^{a}(q),\qquad \mu_{i}^{a}(q)\dot{q}^{i}=0,$$ with $\lambda_a$ denoting the Lagrange multipliers.

The nonholonomic equations of motion are obtained from Lagrange-d'Alembert principle and its local expression is
\begin{equation}\label{LdA:eq}
    \begin{split}
        & \frac{d}{dt}\frac{\partial L}{\partial \dot{q}^{i}} - \frac{\partial L}{\partial q^{i}}= \lambda_{a}\mu^{a}_{i}, \,\,\quad \mu_{i}^{a}(q)\dot{q}^{i}=0
    \end{split}
\end{equation}
where $\lambda_{a}$ is a Lagrange multiplier that might be computed using the constraints.

%\subsection{Nonholonomic systems with inequality constraints}

%\textcolor{red}{hay que unificar la notacion para la ligadura, o usamos $A$ o $\mu$.}

%A nonholonomic system is a mechanical system whose velocities are constrained. We will restrict ourselves to the case of linear constraints on the velocities, where the velocity lies in a subspace of the tangent space. The collection of this subspaces forms a distribution, denoted by $\mathcal{D}$, and is locally given by an expression of the type $\mu^{a}_{i}\dot{q}^{i}=0$.

\subsection{Constrained Hamel's equations}
Consider a nonholonomic system determined by a Lagrangian function $L:TQ\to\mathbb{R}$ and a constraint distribution $\mathcal{D}$. Let $\{u_{1}, \dots, u_{n}\}$ be a local basis of vector fields in $Q$ such that $\mathcal{D}_q=\hbox{span}\{u_1(q),\ldots,u_{k}(q)\}$ with $k = n-m$. Each tangent vector $\dot{q}\in TQ$ can be decomposed as 
$$\dot{q}=\sum_{i=1}^{k} v^{i} u_{i} + \sum_{i=k+1}^{n} v^{i} u_{i}$$
where $\displaystyle{\sum_{i=1}^{k} v^{i} u_{i}}$ is the component of $\dot{q}$ along $\mathcal{D}_q$. We will conveniently denote the first term by $\langle u(q),\dot{q}^{\mathcal{D}}\rangle$ and the second by $\langle u(q),\dot{q}^{\mathcal{U}}\rangle.$

Similarly, each $\alpha\in T^{*}Q$ can be uniquely decomposed as $$\alpha=\langle\alpha_{\mathcal{D}},u^{*}(q)\rangle+\langle \alpha_{\mathcal{U}},u^{*}(q)\rangle,$$ where $\langle\alpha_{\mathcal{D}},u^{*}(q)\rangle$ is the component of $\alpha$ along the dual of $\mathcal{D}_q$ and $u^{*}(q)$ denotes the dual frame of $u(q)$. In particular, the annihilator of $\mathcal{D}$, denoted by $\mathcal{D}^{o}$, is generated by $\{u_{k+1}^{*}, \dots, u_{n}^{*}\}$.

Hence, any vector $v\in \mathcal{D}_{q}$ can be written as 
$$v=\langle u(q),v^{\mathcal{D}}\rangle\hbox{  or  } 0=\langle u(q),v^{\mathcal{U}}\rangle.$$
%and therefore $\delta v=\delta v^{\mathcal{D}}$ or $\delta v^ {\mathcal{U}}=0.$
%and therefore $\delta v=\delta v^{\mathcal{D}}$ or $\delta u^{\mathcal{U}}=0$.

Now, the nonholonomic system can also be obtained from  the constrained Hamel's equations. %$$\left(\frac{d}{dt}\frac{\partial\ell}{\partial v}-\left[v^{\mathcal{D}},\frac{\partial\ell}{\partial v}\right]_{q}^{*}-u[\ell]\right)_{\mathcal{D}},\, \langle u(q),v^{\mathcal{U}}\rangle=0,\,\dot{q}=\langle u(q),v^{\mathcal{D}}\rangle.$$
Letting $\ell(q, v)=L(q, v^{i}u_{i})$ be the local expression of the Lagrangian function with respect to coordinates adapted to the local basis $\{u_{i}\}$, these equations are (locally) given by $$\frac{d}{dt}\frac{\partial\ell}{\partial v^{i}}=c_{ji}^{m}\frac{\partial\ell}{\partial v^{m}}v^{j}+u_{i}[\ell],\quad \dot{q}=v^{i}u_{i}(q),\quad i,j=1,\ldots,k,$$  $v^{a}=0, \ a=k+1,...,n. $

\subsection{Nonholonomic systems with inequality constraints}

If $C$ is an inequality constraint on the nonholonomic system, then Lagrange-d'Alembert equations are still valid in the interior of $C$. However, the jump conditions must now be changed to accommodate the constraints our system has on velocities.

\begin{theorem}
    Let $q:[0,h]\rightarrow Q$ be a nonholonomic trajectory of the nonholonomic system $(\ell,\mathcal{D})$ subjected to the inequality constraint $q(t)\in C$. Suppose that this system has an impact against the boundary $\partial C$ at the time $t_{i}\in [0,h]$. Then the trajectory satisfies Lagrange-d'Alembert equations \eqref{LdA:eq} in the intervals $[0, t_{i}^{-}[$ and $]t_{i}^{+},h]$ and at the impact time $t_{i}$, the following conditions hold:
    \begin{equation}\label{nh:jump:equations}
        \begin{split}
            %& \frac{\partial\ell}{\partial v}|_{t=t_{k}^{+}} - \frac{\partial\ell}{\partial v}|_{t=t_{k}^{-}}  \in N_{C}\cup \mathcal{D}^{o} \\
            & \frac{\partial\ell}{\partial v}|_{t=t_{k}^{+}} - \frac{\partial\ell}{\partial v}|_{t=t_{k}^{-}} = \lambda^{a}u_{a}^{*}(q) + \lambda^{0}dg(q) \\
            & E_{\ell}|_{t=t_{i}^{+}} = E_{\ell}|_{t=t_{i}^{+}} \\
            & \dot{q}(t_{i}^{+}) \in \mathcal{D}_{q(t_{i}^{+})},
        \end{split}
    \end{equation}
    with $a=k+1, \dots, n$ and $\lambda^{0}$, $\lambda^{a}$ are Lagrange multipliers to be determined when solving the jump equations.
\end{theorem}

\begin{proof}
    The Lagrange-d'Alembert principle for systems with impacts is defined on the path space
    $$\Omega = \{ (c,t_{i}) \ | \ c:[0,h]\rightarrow Q \text{ is a smooth curve and } t_{i}\in \mathbb{R}\}.$$

    If the mapping $\mathcal{A}:\Omega \rightarrow \mathbb{R}$ is the action then, the Lagrange,d'Alembert principle states that the derivative of the action should annihilate all variations $(\delta q, \delta t_{i})$ with $\delta q \in \mathcal{D}$, i.e., $\delta q = \sum_{i=1}^{k} w^{i} u_{i}$. Since,
\begin{equation*}
    \begin{split}
        \delta \mathcal{A} = & \int_{0}^{t_{i}^{-}} \left[  c_{ij}^k v^{i}\frac{\partial\ell}{\partial v^{k}}+\psi^{i}_{j}\frac{\partial\ell}{\partial q^{i}} - \frac{d}{dt}\frac{\partial\ell}{\partial v^{j}}\right] w^j \ dt \\&+\int_{t_{i}^{+}}^{h} \left[ c_{ij}^k v^{i}\frac{\partial\ell}{\partial v^{k}}+\psi^{i}_{j}\frac{\partial\ell}{\partial q^{i}} - \frac{d}{dt}\frac{\partial\ell}{\partial v^{j}}\right] w^j \ dt \\
         & - \left[ \frac{\partial\ell}{\partial v^{j}} w^j + \ell \delta t_{k}\right]_{t_{i}^{-}}^{t_{i}^{+}}
    \end{split}
\end{equation*}
the fact that constrained Hamel's equations hold on the intervals $[0, t_{i}^{-}[$ and $]t_{i}^{+},h]$ follows from the application of the fundamental theorem of calculus of variations together with the fact that $\delta q \in \mathcal{D}$. 

The jump condition follows from the fact that $q(t_{k})\in \partial C$ from where

$$w^i(t_{k}) u_i (q(t_{k})) + v^i(t_{k}) u_i(q(t_{k})) \delta t_{k} \in T(\partial C)$$

The variations satisfying the previous equation are spanned by variations $\delta q (t_{i}) \in T(\partial C)$ and $\delta t_{i} = 0$ or $\delta t_{i} = 1$ and $\delta q (t_{i}) = - \dot{q}(t_{i})$. From the latter we immediately deduce that
$$\left[ \frac{\partial\ell}{\partial v} v^i - \ell\right]_{t_{i}^{-}}^{t_{i}^{+}} = 0,$$
which is the energy conservation condition in the jump equations. From $\delta t_{i}=0$, we get that 
$$\frac{\partial\ell}{\partial v}|_{t=t_{k}^{+}} - \frac{\partial\ell}{\partial v}|_{t=t_{k}^{-}},$$
annihilates $\delta q$ if either it is on the normal cone $N_{C}$ or it belongs to the annihilator of the distribution $\mathcal{D}$, since $\delta q$ is in $T(\partial C)\cap \mathcal{D}$. Hence,
$$\frac{\partial\ell}{\partial v}|_{t=t_{k}^{+}} - \frac{\partial\ell}{\partial v}|_{t=t_{k}^{-}} = \lambda^{a}u_{a}^{*}(q) + \lambda^{0}dg(q)$$
where $a=k+1, \dots, n$, $\lambda^{0}$ and $\lambda^{a}$ are Lagrange multipliers to be determined when solving the jump equations. This is precisely the first jump equation. The third one follows from the nonholonomic constraints.
\end{proof}

%\begin{lemma}
 %   The map $\Phi_{q}:\mathcal{D}_{q} \rightarrow \mathcal{D}_{q}$ sending $\dot{q}^{-}\mapsto \dot{q}^{+}$ defined by the jump conditions \eqref{nh:jump:equations} is well-defined.
%\end{lemma}
%\begin{remark}
%    The previous jump equations are in accordance with the equations obtained in \cite{Clark} from Weierstrass-Erdemann conditions for impacts.
%\end{remark}

\section{The Chaplygin sleigh knife edge hitting a boundary}\label{sec5}

The Chaplygin system is a celebrated example of a nonholonomic system. Here we considered the Chaplygin system with knife edge planar coordinates determined by $(x, y)$ and orientation given by $\theta$ whose center of mass coordinates coincide with those from the blade. Under these circumstances, the dynamics is given by the Lagrangian function
\begin{equation*}
    L=\frac{m}{2}\left( \dot{x}^2 + \dot{y}^2\right) + \frac{I}{2}\dot{\theta}^2
\end{equation*}
together with the constraint $\sin \theta \dot{x} = \cos \theta \dot{y}$ generating the distribution
$$\mathcal{D}= \left\langle \left\{ 
\cos \theta \frac{\partial}{\partial x} + \sin \theta \frac{\partial}{\partial y}, \frac{\partial}{\partial \theta} \right\} \right\rangle.$$
Consider the local basis of vector field determined by
$$u_{1}=\cos \theta \frac{\partial}{\partial x} + \sin \theta \frac{\partial}{\partial y}, u_{2} = \frac{\partial}{\partial \theta}$$
$$u_{3} = -\sin \theta \frac{\partial}{\partial x} + \cos \theta \frac{\partial}{\partial y}.$$
The relevant structure functions appearing on Hamel's equations are given by
$[u_{1}, u_{2}] = -u_{3},$
implying $c_{12}^{1}=c_{12}^{2}=0$ and $c_{12}^{3}=-1$.

The Lagrangian function with respect to coordinates adapted to this local frame for $TQ$ takes the expression
$$\ell (q, v) = \frac{m}{2} ((v^{1})^2 + (v^{3})^2) + \frac{I}{2} (v^{2})^2.$$
Therefore, constrained Hamel's equations give
\begin{equation*}
    \begin{split}
         m \dot{v}^{1} &= 0\\
         I \dot{v}^{2} &=  0\\
         v^{3} &= 0 \\
         \dot{q} &= v^{1}u_{1} + v^{2}u_{2}
    \end{split}
\end{equation*}
We will examine Hamel's constrained equations when the knife edge impacts the boundary of the inequality constraint
$$C = \{ (x,y,\theta) \ | \ x^{2} + y^{2} \leqslant 1 \}$$
The jump equations \eqref{nh:jump:equations} at a boundary point $(x, y, \theta) \in \partial C$ are now
\begin{equation*}
    \begin{split}
        m (v^{1, +} - v^{1, -}) = & \lambda_{0} (2x\cos\theta + 2y\sin\theta)\\
        I (v^{2, +} - v^{2, -}) = & 0 \\
        m (v^{3, +} - v^{3, -}) = &  \lambda_{3} + \lambda_{0} (-2x\sin\theta + 2y\cos\theta)  \\
        \frac{m}{2} (v^{1, +})^{2} + \frac{I}{2}(v^{2, +})^{2} = &  \frac{m}{2}(v^{1, -})^{2} + \frac{I}{2}(v^{2, -})^{2} \\
        v^{3, +} = & 0
    \end{split}
\end{equation*}
We may elimate the third and fifth equations so that we end up with the system
\begin{equation*}
    \begin{split}
        m (v^{1, +} - v^{1, -}) = & \lambda_{0} (2x\cos\theta + 2y\sin\theta)\\
        I (v^{2, +} - v^{2, -}) = & 0 \\
        \frac{m}{2} (v^{1, +})^{2} + \frac{I}{2}(v^{2, +})^{2} = &  \frac{m}{2}(v^{1, -})^{2} + \frac{I}{2}(v^{2, -})^{2},
    \end{split}
\end{equation*}
whose admissible solution is $v^{1,+}=-v^{1,-}$ and $v^{2,+}=v^{2,-}$.

Below, we simulate Chaplygin system under this inequality constraint for 400 seconds, using $N=4000$ steps and a time-step of $h=0.1$ (see Figure \ref{fig:my_label}). The exact solution of Hamel's equations is known in this case and it was used to draw the motion. We used physical constant $m=I=1$ and initial points $q_{0}=(0, 0, \pi/2)$ and $v_{0}= (0.1, 0.05)$. We can observe how the velocity $v_{1}(t)$ jumps at each impact with the boundary (Figure \ref{fig:my_label2}) and preservation of energy (Figure \ref{fig:my_label3}).
\begin{figure}[htb!]
    \centering
    \includegraphics[scale=0.5]{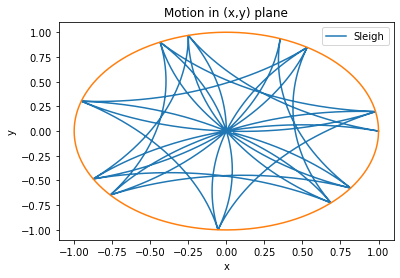}
    \caption{Chaplygin sleigh hitting a circular wall}
    \label{fig:my_label}
\end{figure}

\begin{figure}[htb!]
    \centering
    \includegraphics[scale=0.5]{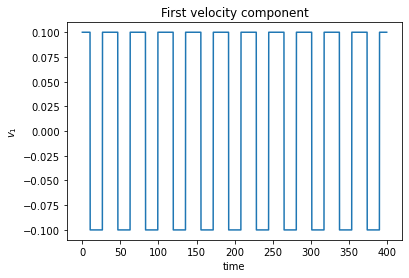}
    \caption{Velocity component $v^{1}(t)$ from Chaplygin sleigh against time.}
    \label{fig:my_label2}
\end{figure}

\begin{figure}[htb!]
    \centering
    \includegraphics[scale=0.5]{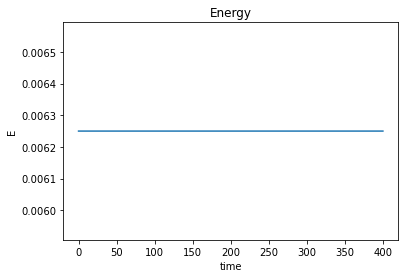}
    \caption{Energy from Chaplygin sleigh in time.}
    \label{fig:my_label3}
\end{figure}

\section{The vertical rolling disk hitting a wall perpendicularly}\label{sec6}

%\subsection{The case of the vertical rolling disk impacting a wall}

The vertical rolling disk is given by the Lagrangian function
\begin{equation*}
    L=\frac{m}{2}\left( \dot{x}^2 + \dot{y}^2\right) + \frac{I}{2}\dot{\theta}^2 + \frac{J}{2}\dot{\varphi}^2
\end{equation*}
together with the non-slipping constraints $\dot{x}=R\dot{\theta}\cos \varphi$, $\dot{y}=R\dot{\theta}\sin \varphi$ generating the distribution
$$\mathcal{D}= \left\langle \left\{ \frac{\partial}{\partial \theta} + R\cos \varphi \frac{\partial}{\partial x} + R \sin \varphi \frac{\partial}{\partial y}, \frac{\partial}{\partial \varphi} \right\} \right\rangle.$$

Consider the local basis of vector field determined by
$$u_{1}=\frac{\partial}{\partial \theta} + R\cos \varphi \frac{\partial}{\partial x} + R \sin \varphi \frac{\partial}{\partial y}, u_{2} = \frac{\partial}{\partial \varphi}$$

$$u_{3} = \frac{\partial}{\partial x} - R\cos \varphi \frac{\partial}{\partial \theta}, u_{4} = \frac{\partial}{\partial y} - R\sin \varphi \frac{\partial}{\partial \theta}$$

The relevant structure functions appearing on Hamel's equations are given by
$$[u_{1}, u_{2}] = R \sin \varphi u_{3} - R \cos \varphi u_{4},$$
implying $c_{12}^{1}=c_{12}^{2}=0$, $c_{12}^{3}=R \sin \varphi$ and $c_{12}^{4}= -R \cos \varphi$.

The Lagrangian function with respect to coordinates adapted to this local frame for $TQ$ takes the expression :
\begin{equation*}
    \begin{split}
        \ell (q, v) = & \frac{m}{2}[(R\cos\varphi v^{1} + v^{3})^{2} + (R\sin\varphi v^{1} + v^{4})^{2}]\\
        & + \frac{I}{2}(v^{1} - R\cos\varphi v^{3} - R \sin\varphi v^{4})^{2}\\
        &+ \frac{J}{2}(v^{2})^{2} 
    \end{split}
\end{equation*}

Taking into account that the partial derivatives evaluated at vectors in $\mathcal{D}$, i.e., at $v^{3}=v^{4}=0$ give
$$\frac{\partial l}{\partial \varphi} = 0,$$
as well as
$$\frac{\partial l}{\partial v^{3}} =   R (m - I ) v^1\cos\varphi $$
and
$$\frac{\partial l}{\partial v^{4}} =  R (m - I ) v^1\sin\varphi $$

The constrained Hamel's equations that now give
\begin{equation*}
    \begin{split}
        & (mR^{2}+I) \dot{v}^{1} = -R \sin \varphi v^{2} \frac{\partial l}{\partial v^{3}} + R \cos \varphi v^{2} \frac{\partial l}{\partial v^{4}} \\
        & J \dot{v}^{2} =  R \sin \varphi v^{2} \frac{\partial l}{\partial v^{3}} - R \cos \varphi v^{2} \frac{\partial l}{\partial v^{4}} + \frac{\partial l}{\partial \varphi}\\
        & v^{3}=v^{4} = 0 \\
        & \dot{q} = v^{1}u_{1} + v^{2}u_{2}
    \end{split}
\end{equation*}
can be simplified, ending up with
\begin{equation*}
    \begin{split}
        & (mR^{2}+I) \dot{v}^{1} = 0 \\
        & J \dot{v}^{2} =  0\\
        & v^{3}=v^{4} = 0 \\
        & \dot{q} = v^{1}u_{1} + v^{2}u_{2}
    \end{split}
\end{equation*}

We will examine Hamel's constrained equations when the disk impacts the boundary of the inequality constraint
$$C = \{ (x,y,\theta,\varphi) \ | \ y + R \sin \varphi \leqslant 10 \}$$
at a constant angle $\varphi = \frac{\pi}{2}$, where the disk makes a right angle with the wall.

In this case, the jump equations \eqref{nh:jump:equations} are simply
\begin{equation*}
    \begin{split}
        (mR^{2} + I) (v^{1, +} - v^{1, -}) = &  \lambda^{0}R(\sin \varphi + \cos\varphi) \\
    J (v^{2, +} - v^{2, -}) = & 0 \\
    R (m - I ) \cos\varphi (v^{1,+} - v^{1, -}) = & \lambda^{3} -
    \lambda^{0} 2 R^{2} \cos^{2}\varphi \\
    R (m - I ) \sin\varphi (v^{1,+} - v^{1, -}) = & \lambda^{4} +
    \lambda^{0} (1 - 2 R^{2} \cos\varphi\sin\varphi) \\
    \frac{mR^{2} + I}{2} & (v^{1, +})^{2} + \frac{J}{2}(v^{2, -})^{2} \\
     = &  \frac{mR^{2} + I}{2}(v^{1, -})^{2} + \frac{J}{2}(v^{2, -})^{2} \\
     v^{3, +} = v^{4, +} = & 0
    \end{split}
\end{equation*}
However, as before, we can eliminte third, fourth and sixth equations to end up simply with
\begin{equation*}
    \begin{split}
        (mR^{2} + I) (v^{1, +} - v^{1, -}) = &  \lambda^{0}R(\sin \varphi + \cos\varphi) \\
    J (v^{2, +} - v^{2, -}) = & 0 \\
    \frac{mR^{2} + I}{2} & (v^{1, +})^{2} + \frac{J}{2}(v^{2, -})^{2} \\
     = &  \frac{mR^{2} + I}{2}(v^{1, -})^{2} + \frac{J}{2}(v^{2, -})^{2} 
    \end{split}
\end{equation*}
In fact, if the impact is perpendicualr to the wall, i.e., $\varphi=\frac{\pi}{2}$, then the unique admissible solution to the jump equations is $v^{1,+}=-v^{1,-}$, $v^{2,+}=v^{2,-}$ and $\lambda_{0} = -\frac{2(mR^{2}+I)v^{1,-}}{R}$.

Finally, we simulate the system under this variational inequality for 18 seconds, using $N=180$ steps and a time-step of $h=0.1$ (see Figure \ref{fig:my_label4}). The exact solution of Hamel's equations is known in this case and it was used to draw the motion. We used physical constant $m=I=J=1$ and initial points $q_{0}=(0, 0, 0, \pi/2)$ and $v_{0}= (1, 0)$. Again, the first velocity component $v^{1}(t)$ is discontinuous (Figure \ref{fig:my_label5}) and energy is preserved along the motion (Figure \ref{fig:my_label6}).
\begin{figure}[htb!]
    \centering
    \includegraphics[scale=0.5]{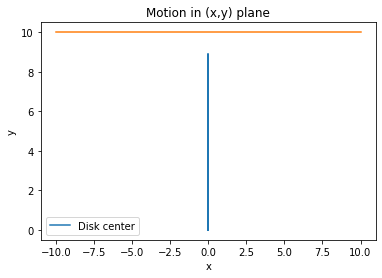}
    \caption{Vertical rolling disk hitting perpendicularly a wall}
    \label{fig:my_label4}
\end{figure}

\begin{figure}[htb!]
    \centering
    \includegraphics[scale=0.5]{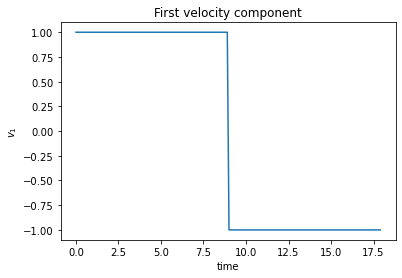}
    \caption{Velocity component $v^{1}(t)$ from rolling disk against time.}
    \label{fig:my_label5}
\end{figure}

\begin{figure}[htb!]
    \centering
    \includegraphics[scale=0.5]{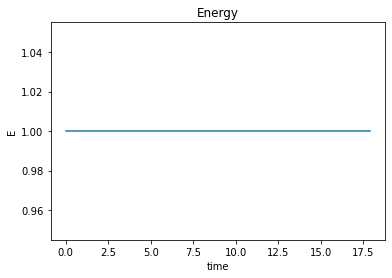}
    \caption{Energy from rolling disk in time.}
    \label{fig:my_label6}
\end{figure}

%\subsection{The case of the rolling ball in a table with four walls}

\end{document}